\documentclass[journal,draftcls,onecolumn,12pt,twoside]{IEEEtran}

\usepackage{amssymb}
\usepackage{amsthm}
\usepackage{mathtools}
\usepackage{amsmath}
\usepackage{graphicx}
\usepackage{url}
\usepackage{amsfonts,amssymb,amsbsy,amsthm}
\usepackage{amsmath}
\usepackage{times}
\usepackage{mathtools}
\usepackage{tikz}
\usepackage{suffix}
\usepackage{mathtools}
\usepackage{rotating}
\usepackage{caption}
\usepackage{subcaption}
\usepackage{graphicx}
\usepackage{subcaption}
\usepackage{mwe}
\usepackage{epstopdf}
\usepackage{mathtools}

\newtheorem{prop}{Proposition}

\newcommand{\alert}[1]{{\color{black}{#1}}}

\newcommand{\erf}{\mathrm{erf}\,}
\newcommand{\erfc}{\mathrm{erfc}\,}
\newcommand{\abs}[1]{\left| #1 \right|}

\makeatletter
\renewcommand*\env@matrix[1][*\c@MaxMatrixCols c]{%
  \hskip -\arraycolsep
  \let\@ifnextchar\new@ifnextchar
  \array{#1}}
\makeatother

\usepackage{cite}
\usepackage{tikz}
\usetikzlibrary{arrows}

\DeclarePairedDelimiterX\MeijerM[3]{\lparen}{\rparen}%
{\begin{matrix}#1 \\ #2\end{matrix}\delimsize\vert\,#3}

\newcommand\MeijerG[8][]{%
  G^{\,#2,#3}_{#4,#5}\MeijerM[#1]{#6}{#7}{#8}}

\WithSuffix\newcommand\MeijerG*[7]{%
  G^{\,#1,#2}_{#3,#4}\MeijerM*{#5}{#6}{#7}}

\begin{document}

\title{Traffic Minimizing Caching and Latent Variable Distributions of Order Statistics}

\author{Joonas~P\"a\"akk\"onen, Prathapasinghe~Dharmawansa, \IEEEmembership{Member,~IEEE}, Ragnar~Freij-Hollanti, Camilla~Hollanti, \IEEEmembership{Member,~IEEE}, and~Olav~Tirkkonen, \IEEEmembership{Member,~IEEE}
\thanks{Joonas P\"a\"akk\"onen and Olav Tirkkonen are with the Department of Communications and Networking, School of Electrical Engineering, Aalto University, Finland (email: \{joonas.paakkonen, olav.tirkkonen\}@aalto.fi).}
\thanks{Ragnar Freij-Hollanti and Camilla Hollanti are with the Department of Mathematics and Systems Analysis, Aalto University, Finland (email: \{ragnar.freij, camilla.hollanti\}@aalto.fi).}
\thanks{Prathapasinghe Dharmawansa is with the Department of Electronic and Telecommunication Engineering, University of Moratuwa, Sri Lanka (email: prathapa@uom.lk).}
\thanks{This work was funded in part by TKK tukis\"a\"ati\"o and the Academy of Finland (grant 299916). The work of C. Hollanti is supported by the Academy of Finland grants \#276031, \#282938, \#303819.}
}

\maketitle

\begin{abstract}
Given a statistical model for the request frequencies and sizes of data objects in a caching system, we derive the probability density of the size of the file that accounts for the largest amount of data traffic. This is equivalent to finding the required size of the cache for a caching placement that maximizes the expected \emph{byte hit ratio} for given file size and popularity distributions. The file that maximizes the expected byte hit ratio is the file for which the product of its size and popularity is the highest -- thus, it is the file that incurs the greatest load on the network. We generalize this theoretical problem to cover factors and addends of arbitrary order statistics for given parent distributions. Further, we study the asymptotic behavior of these distributions. We give several factor and addend densities of widely-used distributions, and verify our results by extensive computer simulations.
\end{abstract}


\IEEEpeerreviewmaketitle

\section{Introduction}
Global IP data traffic is expected to triple over the next five years to more than two zettabytes per year \cite{cisco}. One way to decrease data traffic over the Internet, reduce latencies and server workload is storing popular data close to end-users, which is commonly known as \emph{caching}. The main motivation behind caching data is that one rather stores popular data on intermediate nodes than burdens the data transmission networks. Information centric networks (ICN) and content delivery networks (CDN) are examples of modern paradigms where caching plays a pivotal role. 

While web caching is a rather mature area of research with a history of several decades \cite{belady1966,wangoct1999,podlipnig2003,borst2010}, recently also wireless caching, where contents are stored at either base stations or user equipment themselves, has gained plenty of attention \cite{bastug,gregori,golrezaeibase,maddah,song,ji2016}. Despite the differences in physical architectures, the main questions remain the same. Probably the most fundamental problem of caching is related to cache content replacement. Caches need to update their content due to the dynamic changes of file request frequencies to avoid \emph{cache pollution}, which occurs when unwanted data remain in the cache even though more important data could be cached instead.

The most common metric measuring the performance of a cache is the \emph{hit rate}, which is defined as the ratio between the number of data objects served from the cache and the total number of requests. The drawback of this metric is that a high hit rate does not necessarily imply a high reduction in data traffic. This is the case especially if the cacheable files are of largely varying sizes. Hence, a better metric is the \emph{byte hit rate}, which is defined as the ratio between the amount of data traffic served by the caches and the total amount of data traffic generated by the clients. The higher the byte hit rate, the more traffic is offloaded from the origin server.

There exists a natural tradeoff between the cache size and the amount of data offloaded from the server, which has been studied in, \emph{e.g.}, \cite{tamoorjune2016}. The energy-efficiency of caching placement has been optimized in \cite{osmanmay2011}. Optimal cache sizing has been studied for Peer-to-Peer (P2P) systems with realistic bandwidth and cache storage costs (see, \emph{e.g.}, \cite{zhaioptimaldec2011}). Recently, in \cite{liumay2016}, the authors find the required cache size in backhaul-limited wireless networks, while \cite{zhang2015} employs a scheme where the caching probability increases proportionally with file size. Interestingly, studies on video data show that not all parts of video files are equally popular, so a natural extension of popularity-based caching is caching the most popular parts of data files \cite{maggi}.

Several cache replacement algorithms have been suggested in the literature (see, \emph{e.g.}, \cite{podlipnig2003} for a survey). Some of the most widely-used caching replacement algorithms are Least-Recently-Used (LRU), Least-Frequently-Used (LFU) and SIZE. When an object in the cache is about to be replaced, LRU evicts the objects that were least recently used, LFU evicts the least frequently used objects, and SIZE evicts the largest objects. While these policies are simple and perform well in terms of the hit rate, they perform poorly in terms of the byte hit rate. To address this issue,\cite{cherkasova1998} presents an efficient cache replacement algorithm that takes both the file size and the file request frequencies into account. Since then, several replacement algorithms have been proposed \cite{ali2011,zhang2012,podlipnig2003}. Furthermore, recent research has actively studied coded caching \cite{niesen2014,zhang2015}, and also incorporated machine learning techniques in replacement policies \cite{romano2011,ali2009,abdalla2015}.

A typical feature of data traffic is that the popularity of data objects is highly time-variant, a phenomenon which in the caching literature is known as \emph{temporal locality} \cite{olmos,traverso}. In other words, in a realistic caching system content needs to be constantly updated to keep up with the dynamic behavior of user preferences. Furthermore, a realistic assumption is that the exact sizes and numbers of requests of cacheable data files cannot be known beforehand, unlike models such as the Independent Reference Model (IRM) \cite{fagin} suggest. Nevertheless, past observations can be utilized to estimate file popularities (see, \emph{e.g.}, \cite{bharath} and the references therein).

In this paper, motivated by the aforementioned observations, we study a caching system where the cached files are chosen so that the expected amount of data traffic served by the cache is maximized. Specifically, the files are chosen so that the expected byte hit rate is maximized. We derive the probability density function for the required size of the cache for given file size and popularity distributions in the general case. For instance, our research provides an answer to the following question: How much storage space is needed for the file that generates most data traffic? To the best of our knowledge, this is the first paper that takes such a statistical approach to cache size requirements.

An example of a ranked file catalogue is shown in Table \ref{examp}. The files are ranked according to their bandwidth consumption, \emph{i.e.}, the product of their size and their popularity. This product we call the \emph{importance} of the file. The more bandwidth a file consumes, the more important it is to cache it to minimize the expected backhaul data traffic. We see that, in the example of Table \ref{examp}, neither the largest nor the most popular object is the most important one.

\begin{table}[!h]
\begin{center}
    \begin{tabular}{| c | c | c |}
    \hline
    File size (bits) & File popularity $\left(\text{requests/sec}\right)$ & Bandwidth consumption (bits/sec) \\ \hline
    22 & 15 & 330 \\ \hline
    73 & 2 & 146 \\ \hline
    4 & 36 & 144 \\ \hline
    \end{tabular}
\end{center}
\caption{An example file catalogue of three files. The size of the most important file is 22.}
\label{examp}
\end{table}

Solving the aforementioned cache size problem opens up a new direction of research on latent variables of order statistics. Along this line, we extend our contribution of \emph{factor distributions} to \emph{addend distributions}. Furthermore, we present several examples of factor and addend distributions for widely-used parent distributions. The contributions of this work can be summarized as follows.

\begin{itemize}
  \item We find the probability density function of the $k^{\text{th}}$ smallest file of a file catalogue the file sizes and expected request frequencies of which are drawn from given \alert{independent} probability distributions. This problem is analytically formulated and solved.
  \item We derive a method to find the required cache size for a case where a certain percentage of all cacheable files are required to be cached so that the byte hit rate is maximized.
  \item We derive closed-form addend and factor densities for certain widely-used probability distributions and verify them by computer simulations.
\end{itemize}

The rest of the paper is organized as follows. Section \ref{sec:system} presents the caching system model. Section \ref{analysissec} provides closed-form expressions for the addend and factor densities with concrete examples for widely-used probability distributions. Section \ref{asympsec} investigates the asymptotic behaviour of our results. Section \ref{requcache} solves the cache size problem with illustrative examples. Section \ref{conclusionssec} provides concluding remarks and envisages potential future work.

\section{System Model}\label{sec:system}
Consider a network with users requesting files from a remote server with a cache between the users and the server according to Figure \ref{fig:symo}. The cache contents are updated periodically to accommodate for the temporal locality nature of cacheable files. At each cache content update, the sizes and the popularities of the potentially cacheable files are sampled independently from file size and popularity distributions, respectively. The system designer is only aware of the statistical distributions of file sizes and popularities, and must choose an appropriate cache size given a requirement on either the number of files, or a percentage of files in a file catalogue of $n$ files, that should fit in the cache.

If the file requested by a user is cached, the cache serves the user. However, if the requested file is not cached, the origin server transmits the file and the link between the cache and the server must be used. The cache has a limited storage capacity and stores full data objects so that the expected byte hit rate is maximized. For tractability, we assume perfect knowledge of the popularity of each cacheable file, which we further assume to be static until the next cache content update.

The files are ranked according to their bandwidth consumption, \emph{i.e.}, the product of their size and popularity. This product we call the \emph{importance} of the file. We are interested in, for instance, the size of the most important file. Knowing this size is crucial for designing caching systems and the required size of the caches to minimize the backhaul load. In other words, we are interested in the density of a factor of a product. Henceforth, we call this density the \emph{factor distribution}.

\begin{figure}[!h]
\centering
\includegraphics[scale=.5]{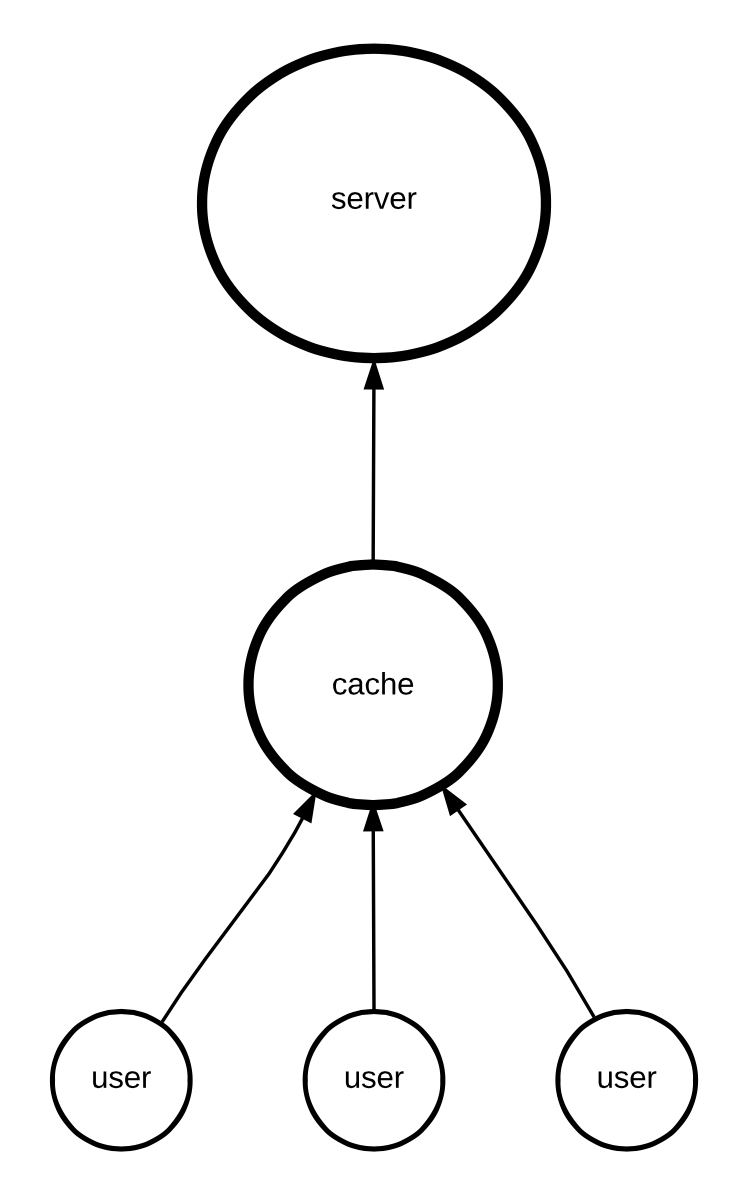}
\caption{Caching model. Files stored on the cache are served by the cache, while files not stored on the cache must be retrieved from the server.}
\label{fig:symo}
\end{figure}

\section{Addend and Factor Distributions of Order Statistics -- General Case}\label{analysissec}
In this section, we show our two main results, \emph{i.e.}, general closed-form expressions for the densities of addends (resp. factors) of the $k^{\text{th}}$ smallest sum (resp. product). Throughout this paper, we call the corresponding distributions \emph{addend} and \emph{factor distributions}, and \alert{we denote the first random variable in a sum or a product as $\tilde{X}$}. \alert{} The $k^{\text{th}}$ smallest random variate is commonly known as the $k^{\text{th}}$ \emph{order statistic}. For further reading, principal references on the theory of order statistics include \cite{clark,david,gungor,wilks}.

Functions of random variables have been excessively studied in the past. Simple examples are the sums and products of independent random variables. While sum distributions are generally direct convolutions of probability density functions, product distributions are often more tedious to derive. Seminal work on product distributions, as well as other functions of random variables, has been done in \cite{donahue1964,rohat1976,springer1979,cook1981}. In this paper, we need to apply both some theory of order statistics and some theory of sum and product distributions to derive addend and factor distributions.

The following example illustrates the main idea of this work for generating random variates of addends. Let $X$ and $Y$ be random variables of a given parent distribution. Now let us generate samples $x_1, x_2, y_1$ and $y_2$ and thence sums $x_1+y_1$ and $x_2+y_2$. Now let us look at the greater of these two sums. If $x_1+y_1>x_2+y_2$, then $x_1$ is a variate of the addend of the maximum of two such sums. Our main goal is to find the density of $X_1$ for which we give a general solution in this section, as well as several examples in Section \ref{sec:addex}.

Let us now introduce some mathematical notation. Let \alert{$\{X_1,X_2,X_3,\dots,X_n$\}} denote \alert{independent} random variables with \alert{densities $\{f_{X_1},f_{X_2},\dots,f_{X_n}\}$, respectively}. Also let there be $n$ sums or products that comprise $m$ random variables each. We call these random variables addends for sums and factors for products. We are interested in the distribution of an individual variable of the $k^{\text{th}}$ largest sum, $S=\sum_{i=1}^m X_i$, the distribution of which we denote $F_S(\cdot)$, or product, $U=\prod_{i=1}^m X_i$, the distribution of which we denote $F_U(\cdot)$. We denote the sums as $S_j = \sum_{i=1}^m X_{j{,}i}$, with $j=1,2,...,n$, and the corresponding $k^{\text{th}}$ order statistic out of $n$ variates as $S_{(k{,}n)}$. For further convenience, let us define new random variables $T = \sum_{i = 2}^m X_i$ and $V=\prod_{i = 2}^m X_i$. Also, let $\mathcal{K}_{n,k} = \frac{n!}{(k{-}1)!(n{-}k)!}$. The newly introduced random variables can be visualized as shown in matrices \eqref{matri1} and \eqref{matri2}:

\begin{align}\label{matri1}
 \begin{bmatrix}[cccc|c]
  X_{1,1} & X_{1,2} & \cdots & X_{1,n} & S_1=\sum_{i=1}^m X_{1,i} \\
  X_{2,1} & X_{2,2} & \cdots & X_{2,n} & S_2=\sum_{i=1}^m X_{2,i} \\
  \vdots  & \vdots  & \ddots & \vdots & \vdots \\
  X_{n,1} & X_{n,2} & \cdots & X_{n,m} & S_n=\sum_{i=1}^m X_{n,i}
 \end{bmatrix}
\end{align}

\begin{align}\label{matri2}
 \begin{bmatrix}[cccc|c]
  X_{1,1} & X_{1,2} & \cdots & X_{1,n} & U_1=\prod_{i=1}^m X_{1,i} \\
  X_{2,1} & X_{2,2} & \cdots & X_{2,n} & U_2=\prod_{i=1}^m X_{2,i} \\
  \vdots  & \vdots  & \ddots & \vdots & \vdots \\
  X_{n,1} & X_{n,2} & \cdots & X_{n,m} & U_n=\prod_{i=1}^m X_{n,i}
 \end{bmatrix}
\end{align}

In general, generating realizations of the addends and factors can be done in the following way: sample \alert{the parent distributions of each addend or factor} to find all $nm$ samples and enter them into an $n\times m$ matrix. In the case of addend distribution, append the matrix with a column with row sums of the random variates as in \eqref{matri1}. In the case of factor distribution, append the matrix with a column with row products of the random variates as in \eqref{matri2}, and then sort the matrix according to the newly generated last column in ascending order so that the smallest value of the column is on top and the largest value is on the bottom.

\alert{Without loss of generality, we now look for the density of the first addend or factor of the individual variates $x$ in the $k^{\text{th}}$ row of the matrix of $n$ rows. Recall that the corresponding random variable $\tilde{X}$ denotes both the addends and the factors. The densities of this random variable are presented in the following propositions, which are the main results of this work.}

\begin{prop}\label{genederi}
The density of the \alert{first} addend of the $k^{\text{th}}$ order statistic of a length-$n$ sample of the sum $S=\sum_{i=1}^m X_i$ is
\begin{align}
\label{eq_prop1}
f_{\tilde{X}}(x) = \mathcal{K}_{n,k} f_{X}(x) \mathbb{E}_{T} \{F_S(x{+}t)^{k-1} (1{-}F_S(x{+}t))^{n-k}\}
\end{align}
where $\mathbb{E}_{T}\{\cdot\}$ denotes expectation with respect to the pdf of $T=\sum_{i=2}^m X_i$.
\end{prop}

\begin{prop}\label{genederip}
The density of the \alert{first} factor of the $k^{\text{th}}$ order statistic of a length-$n$ sample of the product $U=\prod_{i=1}^m X_i$ is
\begin{align}
\label{eq_prop2}
f_{\tilde{X}}(x) = \mathcal{K}_{n,k} f_{X}(x) \mathbb{E}_{V} \{F_U(xv)^{k-1} (1{-}F_U(xv))^{n-k}\}
\end{align}
where $\mathbb{E}_{V}\{\cdot\}$ denotes expectation with respect to the pdf of $V=\prod_{i=2}^m X_i$.
\end{prop}

\begin{proof}
Here we present the proof for the addend distribution. Let $f_B{\left(k{-}1;n{-}1,p\right)}$ be the probability mass function of the binomial distribution representing the event of $k{-}1$ successes out of $n{-}1$ trials with success probability $p$, defined as
\begin{align}
f_B{\left(k{-}1;n{-}1,p\right)} = {n{-}1 \choose k{-}1} p^{k{-}1}(1-p)^{n{-}k}.\nonumber
\end{align}
Without loss of generality, we can focus on the case where the $k^{\text{th}}$ order statistic of a length-$n$ sample of the sum $S=\sum_{i=1}^m X_i$, which we denote $S_{(k,n)}$, takes value $s$. Furthermore, here we denote \alert{the first addend of $s$ as $\tilde{X}$}. Applying the so called \emph{mixed form of Bayes rule} (see, \emph{e.g.}, \cite[2.103a]{mixed}) we obtain

\begin{align*}
&f_{\tilde{X}}(x)
\triangleq f_X(x|S_{(k{,}n)}{=}s) = \frac{\mathbb{P}(S_{(k{,}n)}{=}s|\tilde{X}=x)f_X(x)}{\mathbb{P}(S_{(k{,}n)}{=}s)}.
\end{align*}

Now consider a given variate $x$ and a variate $t$ of $T=\sum_{i=2}^m X_i$. The probability that an arbitrarily chosen sum $s$ is the $k^{\text{th}}$ smallest sum is equal to the probability that $k-1$ out of the remaining $n-1$ sums take value smaller than $x+t$. This probability becomes
\begin{align*}
\mathbb{P}(S_{(k{,}n)}{=}s|\tilde{X}=x)&=\int_{-\infty}^\infty f_B\left(k{-}1;n{-}1,F_S(x{+}t)\right)\mathop{dF_T}\nonumber \\
&=\mathbb{E}_{T} \{f_B\left(k{-}1;n{-}1,F_S(x{+}t)\right)\}\nonumber \\
&={n{-}1 \choose k{-}1} \mathbb{E}_{T} \{F_S(x{+}t)^{k-1} (1{-}F_S(x{+}t))^{n-k}\}.
\end{align*}
The probability that an arbitrarily chosen $s$ is the $k^{\text{th}}$ smallest sum out $n$ sums is $\mathbb{P}(S_{(k{,}n)}{=}s)=\frac{1}{n}$.
Thus, the addend density becomes
\begin{align*}
f_{\tilde{X}}(x)&= \frac{\mathbb{E}_{x,T} \{f_B\left(k{-}1;n{-}1,F_S(x{+}t)\right)\} f_{X}(x)} {\mathbb{E}_S \{f_B\left(k{-}1;n{-}1,F_S(s)\right)\}} \nonumber \\ \nonumber
&= n {n{-}1 \choose k{-}1} \mathbb{E}_{x,T} \{F_S(x{+}t)^{k-1} (1{-}F_S(x{+}t))^{n-k}\} f_{X}(x) \\
&=\mathcal{K}_{n,k} f_{X}(x) \mathbb{E}_{T} \{F_S(x{+}t)^{k-1} (1{-}F_S(x{+}t))^{n-k}\}.
\end{align*}
\end{proof}
The proof of Proposition 2 is nearly identical to that of Proposition 1 and is thus omitted.

Note that by setting $m=1$ in either \eqref{genederi} or \eqref{genederip}, we recover the well-known density of the $k^{\text{th}}$ order statistic of random variable $X$ (see, \emph{e.g.}, \cite{david})
\begin{align}\label{wellknown}
f_{X_{(k{,}n)}}(x)= \frac{n!}{(k-1)!(n-k)!} \left(F_X(x)\right)^{k-1} \left(1{-}F_X(x)\right)^{n-k} f_{X}(x).
\end{align}

Similar expressions can be found for more general cases, where the sums or products are not identically distributed, or where not all sums or products comprise the same number of addends or factors. However, in this article, we concentrate only on the symmetric case, where all $n$ sums (resp. products) have exactly $m$ addends (resp. factors). Furthermore, we only demonstrate examples where all the addends or factors are \alert{independent} random variables.

In the remainder of this section, we use Propositions $1$ and $2$ to derive closed-form expressions for widely-used parent distributions. In the examples, all the sums and products have two addends or factors ($m=2$), and there are two sums or products ($n=2$).

With $n=2, k=2,$ and $m=2$, Propositions $1$ and $2$ yield

\begin{align}
\label{addsim}
f_{\tilde{X}}(x) &= 2f_{X}(x) \mathbb{E}_{T} \{F_S(x{+}t)\}
\end{align}
and
\begin{align}
\label{facsim}
f_{\tilde{X}}(x) &= 2 f_{X}(x) \mathbb{E}_{V} \{F_U(xv)\},
\end{align}
respectively.

\subsection{Addend Distributions -- Examples}\label{sec:addex}
\subsubsection{Uniform Addend Distribution}
Let $X \sim \text{U}(a,b)$. Thus, $f_X(x)=1,\text{ when }a\leq{x}\leq b$. Let $S$ follow the distribution of the sum of two standard uniform variables. It is well-known that $S$ follows the triangular distribution, which can be verified by convolution: $f_S(x) = \int_{-\infty}^{\infty}f_X(\tau)f_X(x{-}\tau)\mathop{d\tau}$. \alert{Now by using \eqref{addsim}, the uniform addend density yields
\begin{align*}
f_{\tilde{X}}(x)&=
\frac{-2}{3 (a-b)^4}x^3 +  \frac{a+b}{(a-b)^4}x^2 + \frac{a^2-4 a b+b^2}{(a-b)^4}x \\
&\quad + \frac{-5 a^3+12 a^2 b-6 a b^2+b^3}{3 (a-b)^4}, \quad a \leq x \leq b
,
\end{align*}
which is plotted in Fig. \ref{uniformmax22sumfig} for $a=0$ and $b=1$.
}

\subsubsection{Normal Addend Distribution}
The normal pdf is $f_X(x)=(2\pi\sigma^2)^{-\frac{1}{2}} \exp{\left({-\frac{(\mu-x)^2}{2 \sigma^2}}\right)}$, and the sum of two independent normally distributed random variables has distribution $\mathcal{N}(2\mu,2\sigma^2)$. Moreover, if $S=X_{1,1}+X_{1,2}\sim\mathcal{N}(2\mu,2\sigma^2) $ and $T=X_{2,2}\sim\mathcal{N}(\mu,\sigma^2)$, then $W=S-T=X_{1,1}+X_{1,2}-X_{2,2}\sim\mathcal{N}(\mu,3\sigma^2)$.
We can therefore write \[ \mathbb{E}_{T} \{F_S(x{+}t)\}=\mathbb{P}(S<x+T)=\mathbb{P}(W<x).\]
The cumulative distribution function $W$ is $F_W(x)=\frac{1}{2}\left(1+\erf\left(\frac{x-\mu}{{{\sqrt{6}}\sigma}}\right)\right)$, where $$\erf(x)=\frac{2}{\sqrt{\pi}}\int_{0}^x\exp{\left(-t^2\right)}\mathop{dt}$$ is the error function.

The normal addend density thus becomes
\begin{align}\label{normaddpdf}
f_{\tilde{X}}(x)
&=({2 \pi } \sigma^2)^{-\frac{1}{2}}\exp{\left(-\frac{(\mu-x)^2}{2 \sigma^2}\right)} \erfc\left(\frac{\mu-x}{\sqrt{6\sigma^2}}\right),
\end{align}
where $\erfc(x)=1-\erf(x)$ is the complementary error function. The density \eqref{normaddpdf} is plotted in Fig. \ref{normmax22sumfig} for the standard normal parent distribution.

\subsubsection{Exponential Addend Distribution}
It is well-known that the sum of two exponential random variables, with densities $f_X(x)=\lambda \exp{\left(-\lambda x\right)}$ each, follows the Erlang distribution with density $f_S(x)=\lambda x \exp{\left(-\lambda x\right)}$. Thus, the corresponding cumulative distribution function is $F_{\text{S}}(x)=1-\lambda x\exp{\left(-\lambda x\right)} - \exp{\left(-\lambda x\right)}.$

\alert{
The exponential addend density becomes
\begin{align}
f_{\tilde{X}}(x)
&=\lambda \exp{\left(-\lambda x\right)} \left(2-\exp{\left(-\lambda x\right)} \left(\lambda x+\frac{3}{2}\right)\right),
\end{align}
which is plotted in Fig. \ref{expmax22sumfig} for $\lambda=1$.
}

\subsubsection{Rayleigh Addend Distribution}
The Rayleigh pdf is
\begin{align}
\label{Rayleigh}
f(x)=\frac{2x}{\sigma}\exp{\left(-\frac{x^2}{\sigma}\right)}
\end{align}
and the density of the sum of two iid Rayleigh random variables can be expressed as \cite{PD19}[{\bf Stuber}]
\begin{align}
\label{Rayleigh_Sum}
f_S(x)=\frac x\sigma \exp{\left(\frac{-x^2}{\sigma}\right)}+\sqrt{\frac{\pi}{2\sigma}}\exp{\left(-\frac{x^2}{2\sigma}\right)}\erf{\left(\frac{x}{\sqrt{2\sigma}}\right)}\left(\frac{x^2}{\sigma}-1\right).
\end{align}
Therefore, following (\ref{addsim}), we can write the Rayleigh addend density as
\begin{align}
\label{Rayleigh_addend}
f_{\tilde{X}}(x)=\frac{4x}{\sigma}\exp{\left(-\frac{x^2}{\sigma}\right)} \int_0^\infty \int_0^{x+t} 
f_S(z) f(t){\rm d} z {\rm d} t.
\end{align}
Now we use (\ref{Rayleigh}) and (\ref{Rayleigh_Sum}) in (\ref{Rayleigh_addend}) with some algebraic manipulation to obtain the Rayleigh addend density
\begin{align*}
f_{\tilde{X}}(x)=
\frac{4x}{\sigma}\exp{\left(-\frac{x^2}{\sigma}\right)} h(x)
\end{align*}
where
\begin{align}\label{hjuttu}
h(x)=\frac12-\frac12\mathcal{I}(x)+\sqrt{\frac\pi{2\sigma}}\mathcal{J}(x)
\end{align}
in which
\begin{align*}
\mathcal{I}(x)=\frac12\exp\left(-\frac{x^2}{\sigma}\right)-x\exp\left(-\frac{x^2}{2\sigma}\right)\sqrt{\frac{\pi}{2\sigma}}Q\left(\frac{x}{\sqrt{\sigma}}\right)
\end{align*}
where
\begin{align*}
Q\left(x\right) = \frac{1}{\sqrt{2\pi}}\int_x^\infty \exp\left(-\frac{u^2}{2}\right)\mathop{du}
\end{align*}
is the Q-function and
\begin{align*}
\mathcal{J}(x)=\int_0^\infty \int_0^{x+t}
\frac{2t}{\sigma} \exp\left(-\frac{t^2}{\sigma}-\frac{z^2}{2\sigma}\right)
\left(\frac{z^2}{\sigma}-1\right) \text{erf}\left(\frac{z}{\sqrt{2\sigma}}\right) {\rm d}z {\rm d}t.
\end{align*}
Changing the order of integration and changing to polar coordinates, the integrations may be performed, leading to
\begin{align*}
\mathcal{J}(x)=\sqrt{\frac{\sigma}{2\pi}}-\frac{4x}{9}\exp\left(-\frac{x^2}{2\sigma}\right)+\frac{23x}{18}\exp\left(-\frac{x^2}{2\sigma}\right)Q\left(\frac{x}{\sqrt{\sigma}}\right)-\frac{5}{6}\sqrt{\frac{\sigma}{2\pi}}\exp\left(-\frac{x^2}{\sigma}\right)\\+\sqrt{\frac{2\sigma}{3\pi}}\left(\frac{4x^2}{9\sigma}-\frac{2}{3}\right)\exp\left(-\frac{x^2}{3\sigma}\right)\phi(x)
\end{align*}
 where
 \begin{align*}
 \phi(x)=\pi Q\left(\frac{x}{\sqrt{3\sigma}}\right)-2\int_0^{\frac{\pi}{6}}\exp\left(-\frac{x^2}{6\sigma \sin^2 \theta}\right)\mathop{d\theta}.
 \end{align*}
While the above single integral does not have a closed form solution, it can be evaluated numerically. The density is plotted in Fig. \ref{raymax22sumfig} for $\sigma=1$.

\subsection{Factor Distributions -- Examples}
\subsubsection{Uniform Factor Distribution}

The standard uniform pdf is $f_X(x)=1, \text{when } 0\leq x \leq 1$. The density of the product of two standard uniform random variables is $f_U(x)=-\log x$ \cite{weissuni}, where $\log(\cdot)$ is the natural logarithm. This yields cumulative distribution function $F_U(x)=x-x\log x$.

The uniform factor density becomes
\begin{align*}
f_{\tilde{X}}(x)
=2\int_{0}^{1}\left(-xz\log{xz}+xz\right)\mathop{dz}=-x\log{x}+\frac{3}{2}x,
\end{align*}
which is plotted in Fig. \ref{uniformmax22prodfig}.

\subsubsection{Normal Factor Distribution}
For simplicity, here we only focus on $\mathcal{N}(0,\sigma^2)$, the pdf of which is $\phi(x)=(2\pi\sigma^2)^{-\frac{1}{2}} \exp{\left({-\frac{x^2}{2 \sigma^2}}\right)}$. The product of two iid normal random variables that follow $\mathcal{N}(0,\sigma^2)$ is $f_V(x)=\frac{1}{\pi \sigma^2}K_0\left(\frac{\abs {x}}{\sigma^2}\right)$ \cite{weiss}, where $K_0(\cdot)$ is a modified Bessel function of the second kind of order 0. The normal factor density thus becomes
\begin{align*}
f_{\tilde{X}}(x) &= 2f_X(x) \mathbb{E}_{V} \{F_U(xv)\}=
 2\phi(x)\int_{-\infty}^\infty \int_{-\infty}^{xv} \frac{1}{\pi \sigma^2}K_0\left(\frac{\abs {z}}{\sigma^2}\right)\mathop{dz}\phi(v)\mathop{dv}=
\phi{(x)},
\end{align*}
which is plotted in Fig. \ref{normmax22prodfig}. A detailed proof is shown in Appendix \ref{norfac}.

\subsubsection{Exponential Factor Distribution}\label{sec:expprod}
The exponential distribution with rate $\lambda$ has density $f_X(x)=\lambda \exp{(-\lambda x)}$. If $X$ and $Y$ are exponential with the same rate, and $x,y\in\mathbb{R}$, then it is well-known that $\mathbb{P}(x\cdot X<y\cdot Y)=\frac{x}{x + y}$.
From this observation, and by using the substitution $t=\lambda(x+y)$, the exponential factor density yields
\begin{align}\label{expprodpdf}
f_{\tilde{X}}(x)
&= 2\lambda \exp{(-\lambda x)} \mathbb{P}(X_{1,1}X_{1,2} <X_{2,1}\cdot x)\nonumber\\
&=2\lambda^2x E_1(\lambda x),
\end{align}
where $E_1(\cdot)$ is the exponential integral defined as $E_1(x) = \int_{x}^{\infty}\exp(-t)/t\mathop{dt}$. The density \eqref{expprodpdf} is plotted in Fig. \ref{expmax22prodfig} for $\lambda=1$.

\subsubsection{Rayleigh Factor Distribution}
The Rayleigh distribution as defined earlier is related to the exponential distribution so that if $Y$ is exponential with parameter $\lambda$, then $X=\sqrt{Y\sigma\lambda}$ is Rayleigh with parameter $\sigma$. Thus, utilizing \eqref{expprodpdf} and the fact that $F'(ax^2)=2axf(ax^2)$, we have
\begin{align*}
f_{\tilde{X}}(x)={\frac{{4x^3}}{\sigma^2}} E_1\left(\frac{x^2}{\sigma}\right),
\end{align*}
which is plotted in Fig. \ref{rayleighmax22prodfig} for $\sigma=1$.

Table \ref{tabrefactors} summarizes the addend and factor distributions derived in the last two sections.
\begin{table}[!h]
\begin{center}
    \begin{tabular}{| c | c | c | c | c |}
    \hline
    Distribution &
    pdf of addend of max. sum $f_{\tilde{X}}$ &
    pdf of factor of max. product $f_{\tilde{X}}$\\
    \hline
    General &
    $2f_X(x) \mathbb{E}_{T} \{F_S(x+t)\}$ &
    $2f_X(x) \mathbb{E}_{V} \{F_U(xv)\}$\\
    \hline
    
    \text{U}(a,b) &
    $\frac{-2}{3 (a-b)^4}x^3 +  \frac{a+b}{(a-b)^4}x^2 + \frac{a^2-4 a b+b^2}{(a-b)^4}x + \frac{-5 a^3+12 a^2 b-6 a b^2+b^3}{3 (a-b)^4}$ &
    $-x\log{x}+\frac{3}{2}x\quad (a=0$, $b=1)$\\
    \hline
    
    $\mathcal{N}(\mu,\sigma^2)$ &
    $({2 \pi } \sigma^2)^{-\frac{1}{2}}\exp{\left(-\frac{(\mu-x)^2}{2 \sigma^2}\right)} \erfc\left(\frac{\mu-x}{\sqrt{6\sigma^2}}\right)$ &
    $(2\pi\sigma^2)^{-\frac{1}{2}} \exp{\left(-\frac{x^2}{2\sigma^2}\right)}\quad (\mu=0)$\\
    \hline
    
    Exp($\lambda$) &
    $\lambda \exp{(-\lambda x)} \left(2-\exp{(-\lambda x)} (\lambda x+\frac{3}{2})\right)$ &
    $2\lambda^2xE_1(\lambda x)$\\
    \hline
    
    Rayleigh($\sigma$) &
    $\frac{4x}{\sigma}\exp{\left(-\frac{x^2}{\sigma}\right)}h(x)$ &
    $\frac{{x^3}}{\sigma^{4}} E_1\left(\frac{x^2}{2 \sigma^2}\right)$\\
    \hline
    \end{tabular}
\end{center}
\caption{Addend and factor densities of the general form for iid random variables for select distributions  with $m=n=k=2$. Here $h(x)$ is defined as in \eqref{hjuttu}.}
\label{tabrefactors}
\end{table}

Figures \ref{addendfigs} and \ref{factordistrs} plot the densities of Table \ref{tabrefactors} with corresponding numerical histograms.
\begin{figure*}
\centering
\begin{subfigure}[b]{0.475\textwidth}
\centering
\includegraphics[width=\textwidth]{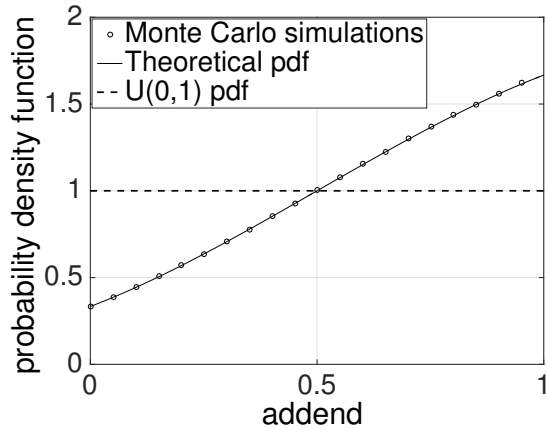}
\caption[]%
            {{\small Density of an addend of the maximum of two sums of two standard uniform $\text{U}(0,1)$ variables.}}    
\label{uniformmax22sumfig}
\end{subfigure}%
\hfill
\begin{subfigure}[b]{0.475\textwidth}  
\centering
\includegraphics[width=\textwidth]{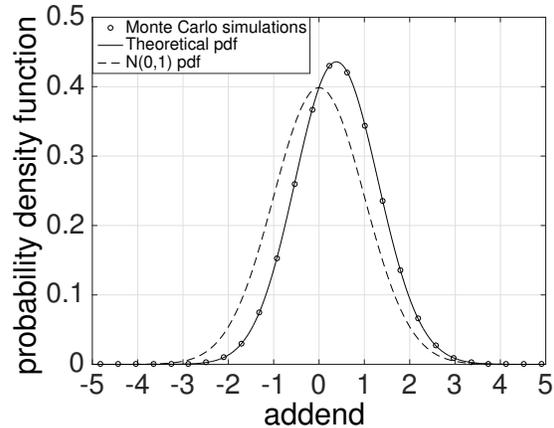}
\caption[]%
            {{\small Density of an addend of the max. of two sums of two standard normal $\mathcal{N}(0,1)$ variables.}}    
\label{normmax22sumfig}
\end{subfigure}
\vskip\baselineskip
\begin{subfigure}[b]{0.475\textwidth}  
\centering
\includegraphics[width=\textwidth]{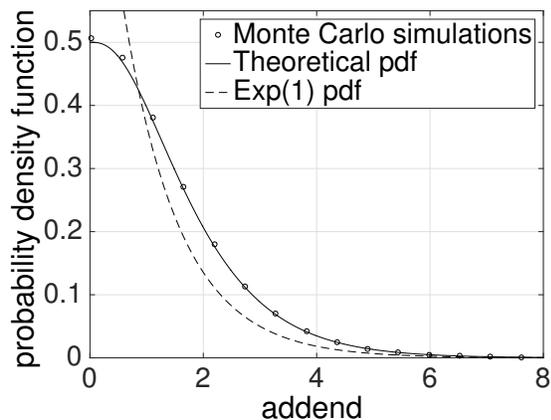}
\caption[]%
            {{\small Density of an addend of the maximum of two sums of two standard exponential Exp(1) variables.}}    
\label{expmax22sumfig}
\end{subfigure}
\quad
\begin{subfigure}[b]{0.475\textwidth}   
\centering
\includegraphics[width=\textwidth]{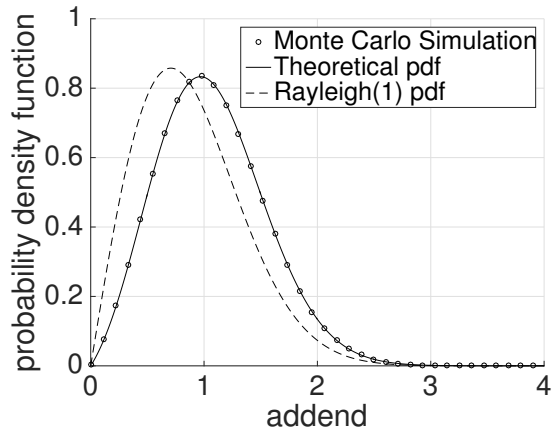}
\caption[]%
            {{\small Density of an addend of the maximum of two sums of two standard Rayleigh(1) variables.}}    
\label{raymax22sumfig}
\end{subfigure}

\caption[  ]
        {\small Select addend distributions with numerical verifications.} 
\label{addendfigs}
\end{figure*}

    \begin{figure*}
        \centering
        \begin{subfigure}[b]{0.475\textwidth}
            \centering
            \includegraphics[width=\textwidth]{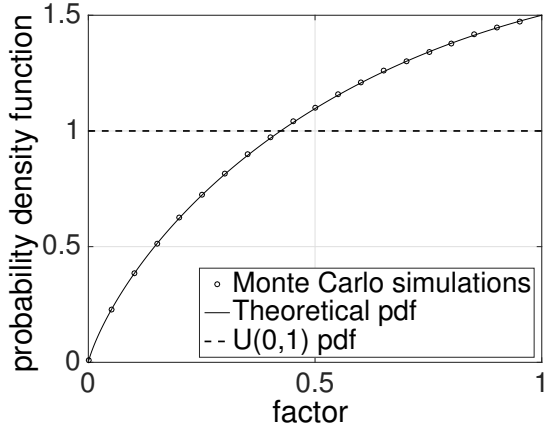}
            \caption[]%
            {{\small Density of a factor of the maximum of two products of two standard uniform $\text{U}(0,1)$ variables.}}    
            \label{uniformmax22prodfig}
        \end{subfigure}
        \hfill
        \begin{subfigure}[b]{0.475\textwidth}  
            \centering 
            \includegraphics[width=\textwidth]{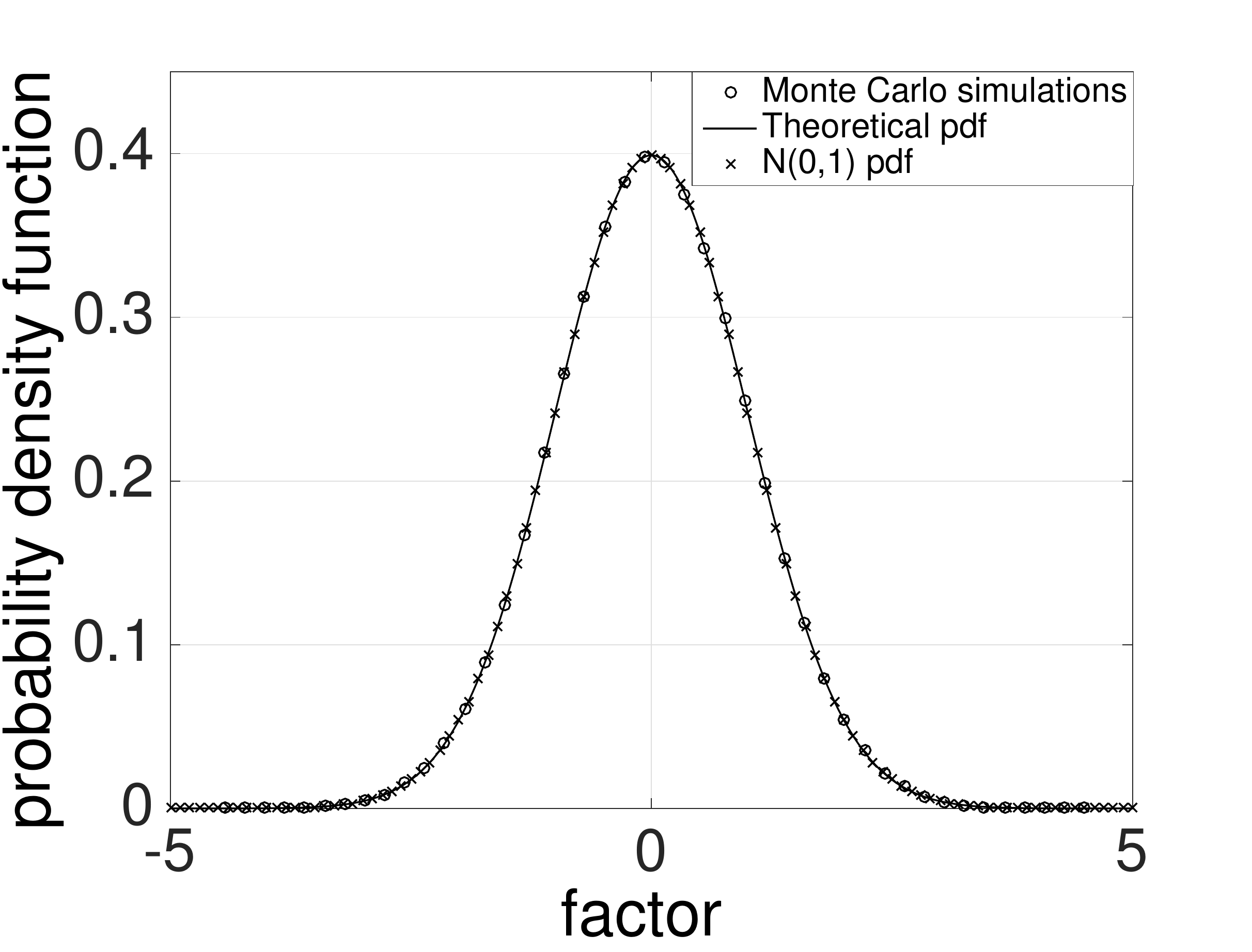}
            \caption[]%
            {{\small Density of a factor of the maximum of two products of two standard normal $\mathcal{N}(0,1)$ variables.}}    
            \label{normmax22prodfig}
        \end{subfigure}
        \vskip\baselineskip
        \begin{subfigure}[b]{0.475\textwidth}   
            \centering 
            \includegraphics[width=\textwidth]{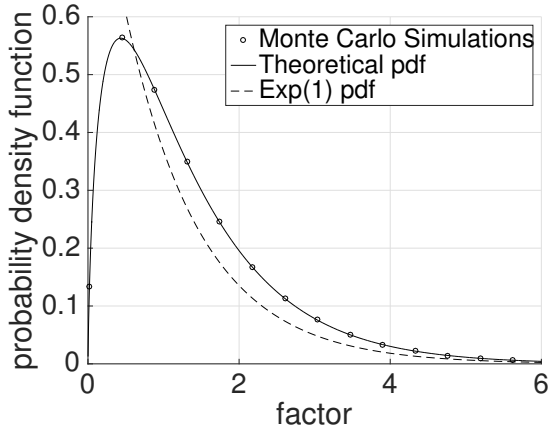}
            \caption[]%
            {{\small Density of a factor of the maximum of two products of two standard exponential Exp(1) variables.}}    
            \label{expmax22prodfig}
        \end{subfigure}
        \quad
        \begin{subfigure}[b]{0.475\textwidth}   
            \centering 
            \includegraphics[width=\textwidth]{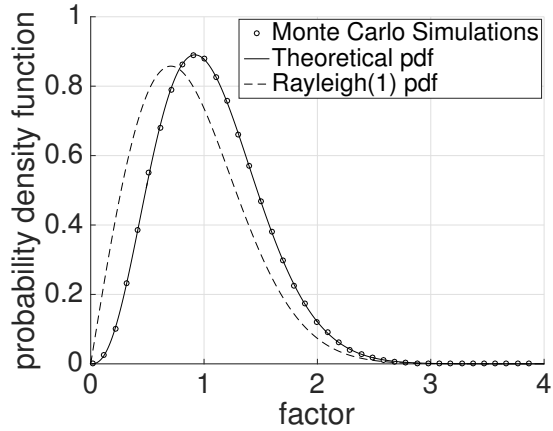}
            \caption[]%
            {{\small Density of a factor of the maximum of two products of two standard Rayleigh(1) variables.}}    
            \label{rayleighmax22prodfig}
        \end{subfigure}
        \caption[  ]
        {\small Select factor distributions with numerical verifications.} 
        \label{factordistrs}
    \end{figure*}
\newpage
\section{Asymptotic Behavior}\label{asympsec}
In general, it appears difficult to characterize the impact of $n$ and $k$ on the addend or factor distribution.  To circumvent this difficulty, here we focus on the asymptotic regime. In particular, we investigate the double scaling limit $n,k \to \infty$ such that the ratio $k/n$ remains fixed. Since the direct evaluation of such limits using either
(\ref{eq_prop1}) or (\ref{eq_prop2}) seems an arduous task, in what follows, we adopt stochastic convergence concepts detailed in \cite{Chung}.

\newpage
Before proceeding, a few notations are in order. \alert{Let \alert{$\{X_1,X_2,X_3,\dots,X_n$\}} denote \alert{independent} random variables with \alert{densities $\{f_{X_1},f_{X_2},\dots,f_{X_n}\}$, respectively. We are interested in finding the limiting distribution of $\tilde{X}=X_1$.}} Let $g_m(\cdot)$ denote the density of $T_1=\sum_{j=1}^{m}X_j$ and, similarly, $g_{m-1}(\cdot)$ denote the density of $T_2=\sum_{j=2}^{m}X_j$, while $G_m(\cdot)$ denotes the corresponding cumulative distribution function and $G_m^{-1}(\cdot)$ its inverse. Now we have the following key result: 

\newpage
\begin{prop}\label{dsl}
Let $f_X$ be bounded and have an infinite support\footnote{This particular condition can easily be relaxed to accommodate single-sided probability density functions.}. Then as $n\to\infty$ such that $k/n\to \eta$ with $0<\eta<1$, \alert{the asymptotic addend distribution becomes}
\begin{align}
\lim\limits_{k/n=\eta, n\to\infty} f_{\tilde{X}}(x) = \frac{g_{m{-}1}\left(\beta{-}x\right)}{g_{m}\left(\beta\right)}f_X(x),\end{align}
\end{prop}
where $\beta = G_m^{-1}(\eta)$.

The proof is shown in Appendix \ref{dslproof}.

As a particular example, for standard normal with $m=2$, we obtain
\begin{align*}
g_1(x) = \phi(x) = (2\pi)^{-\frac{1}{2}} \exp{\left(-\frac{1}{2}x^2\right)},\quad
g_2(x) = (4\pi)^{-\frac{1}{2}} \exp{\left(-\frac{1}{4}x^2\right)},
\end{align*}
which yields
\begin{align*}
\lim\limits_{k/n=\eta, n\to\infty} f_{\tilde{X}}(x) = \pi^{-\frac{1}{2}} \exp{\left(-\left(x-\frac{1}{2}\beta\right)^2\right)},
\end{align*}
where $\beta = \Phi_2^{-1}(\eta)$ is the inverse of the cumulative distribution function of $\mathcal{N}(0,2)$ evaluated at $\eta=k/n$. This implies that the limiting distribution is a normal distribution, namely $\mathcal{N}(\frac{1}{2}\beta,\frac{1}{2})$. 

To preserve the notation for the case of factor distribution, let $g_m(\cdot)$ denote the density of $V_1=\prod_{j=1}^{m}X_j$ and, similarly, $g_{m-1}(\cdot)$ denote the density of $V_2=\prod_{j=2}^{m}X_j$, while $G_m(\cdot)$ denotes the corresponding cumulative distribution function and $G_m^{-1}(\cdot)$ its inverse.

The double scaling limit of the factor distribution is given by the following result.\\
\begin{prop}\label{dslp}
Let $f_X$ be bounded. Then as $n\to\infty$ such that $k/n\to \eta$ with $0<\eta<1$, \alert{the asymptotic factor distribution becomes}
\begin{align}\label{afd}
\lim\limits_{k/n=\eta, n\to\infty} f_{\tilde{X}}(x) = \frac{g_{m{-}1}\left(\beta/x\right)}{g_{m}\left(\beta\right)|x|}f_X(x),
\end{align}
\end{prop}
where $\beta=G_m^{-1}(\eta)$. 

The proof is shown in Section \ref{dslproof}. For standard normal with $m=2$, we obtain
\begin{align*}
g_1(x) = \phi(x) = (2\pi)^{-\frac{1}{2}} \exp{\left(-\frac{1}{2}x^2\right)},\quad
g_2(x) = \frac{1}{\pi}K_0(\abs x),
\end{align*}
which yields
\begin{align*}
\lim\limits_{k/n=\eta, n\to\infty} f_{\tilde{X}}(x;\beta) = \frac{1}{2\abs x K_0(|\beta|)}\exp{\left(-\frac{1}{2}(x^2+\beta^2/x^2)\right)},
\end{align*}
where $$\eta = \eta(\beta) \coloneqq \int_{-\infty}^{\beta}\frac{1}{\pi}K_0(\abs z)\mathop{dz}.$$

Knowing the double scaling limits provides an effective tool for approximating the corresponding addend and factor distributions. In Figures \ref{dslsum} and \ref{dslprod} we see that the approximations hold for relatively small values of $k$ and $n$ when the parent distribution is standard normal.

\begin{figure}[!h]
    \centering
    \begin{subfigure}[h]{0.5\textwidth}
        \centering
        \includegraphics[height=2.6in]{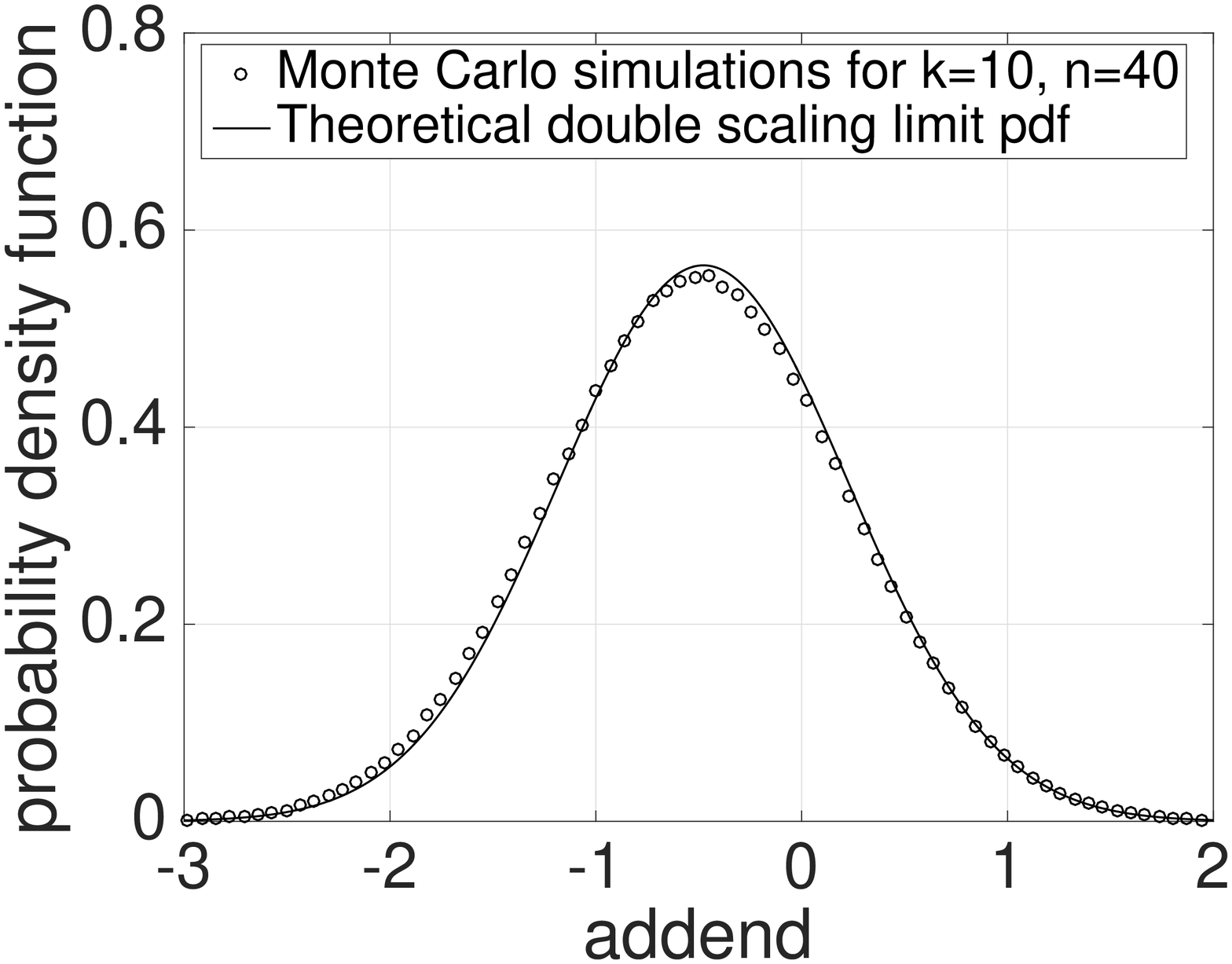}
\caption{Double scaling limit of the addend distribution.}
\label{dslsum}
    \end{subfigure}%
    ~ 
    \begin{subfigure}[h]{0.5\textwidth}
        \centering
        \includegraphics[height=2.6in]{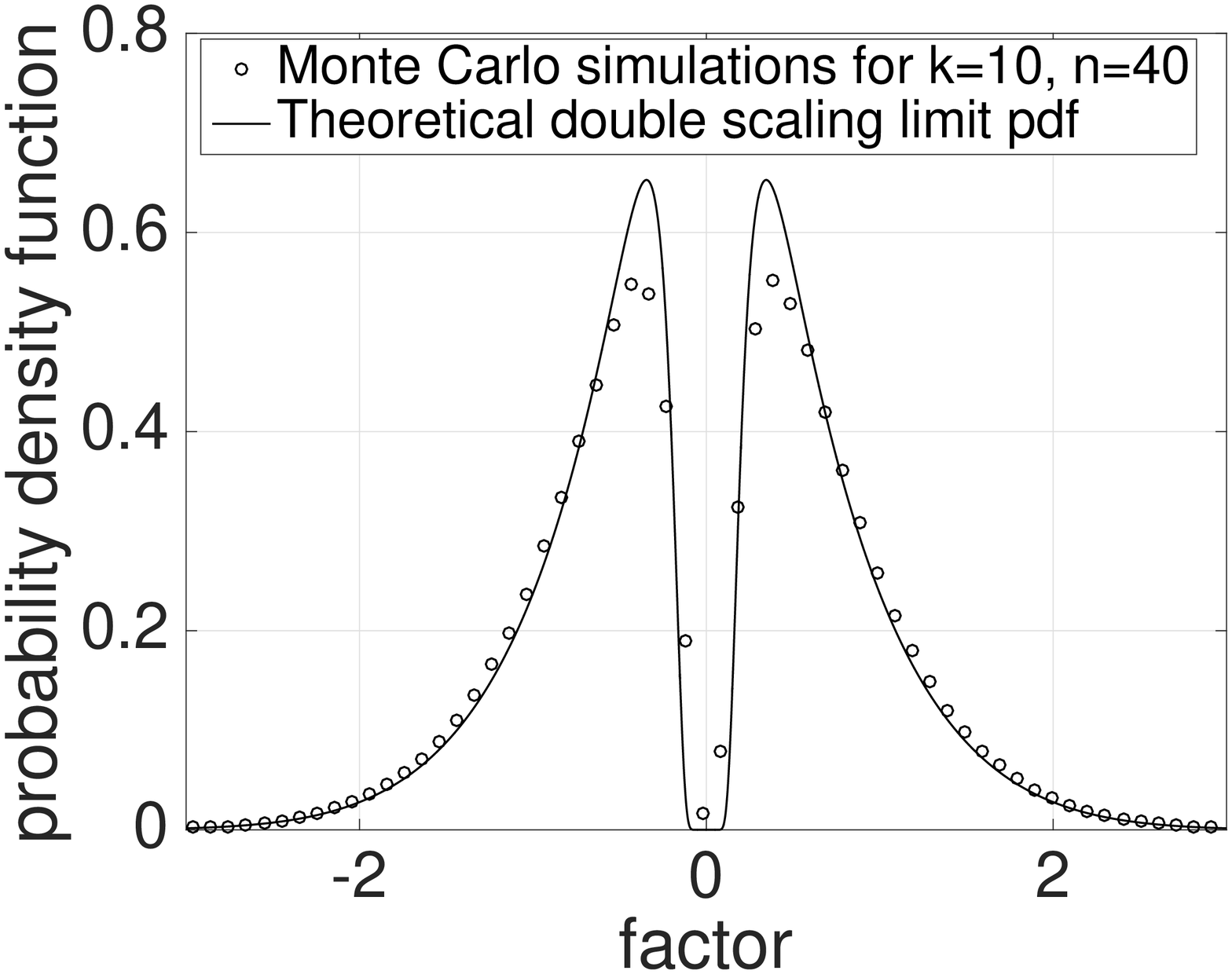}
\caption{Double scaling limit of the factor distribution.}
\label{dslprod}
    \end{subfigure}
\caption{Examples of the double scaling limit densities of the addend and factor distributions when the parent distribution is $\mathcal{N}(0,1)$. Here $k=10$ and $n=40$. The double scaling limit density with $\eta=10/40=0.25$ approximates the numerical results with relatively good accuracy.}
\label{dsls}
\end{figure}

\section{Examples of Required Cache Sizes}\label{requcache}
Here we return to the original problem of cache sizing and solve two example problems with the newly obtained mathematical tools.

\subsection{Size of the Most Important File}
Consider a file catalogue with $n$ files, the sizes of which are sampled from Gamma$(2,1)$ and the popularities of which are sampled from the standard uniform distribution $\text{U}(0,1)$. The file size is measured in bits and the popularity is measured in requests per second. We are interested in the distribution of the size of the file that generates most traffic, \emph{i.e.}, the file for which the product of its size and popularity, in bits per second, is largest. In other words, we have a case of the factor distribution with $k=n$. Now $m=2$ and the factors are independently but not identically distributed.


It has been shown that the product of Gamma$(2,1)$, with density $f_X(x)=x\exp{(-x)}$, and $\text{U}(0,1)$ distributed random variables follows Exp$(1)$ \cite{johnson}, the cumulative distribution function of which is $F_U(x)=1-\exp(-x)$. \alert{Thus, we now have a case where $U$ follows Exp$(1)$ and $V$ follows $\text{U}(0,1)$ in \eqref{facsim}}, and the size distribution of the file that generates most traffic becomes
\alert{
$$
f_{\tilde{X}}(x) = n e^{-x} (x-H_{n-1}) + \sum_{j=2}^n (-1)^j {n\choose j} \frac{j}{j-1} e^{-i x}~,
$$
where the harmonic numbers are $H_j = \sum_{i=1}^j \frac{1}{i}$. This is plotted in Fig. \ref{gamma1000fig} for $n=1000$. The expected size of the file generating most traffic is then
\begin{align}\label{maxtraffic}
\mathcal{E}_\text{max}(n) = 1 + \frac{1}{n} + H_ {n-1},
\end{align}
which is logarithmic in $n$ and is lower and upper bounded by $1+ \frac{1}{n} + \gamma + \{\log (n-1), \log (n)\}$, where $\gamma$ is the Euler-Mascheroni constant.}

\begin{figure}[!h]
\centering
\includegraphics[scale=.4]{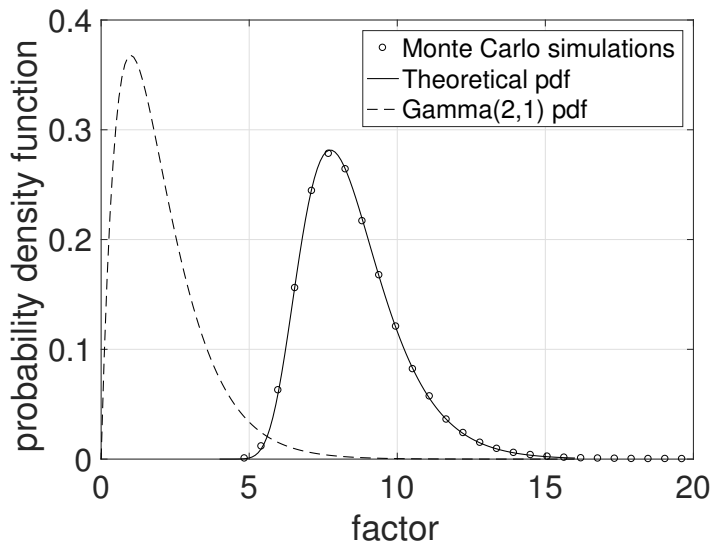}
\caption{\alert{Factor density for parent distribution Gamma(2,1) for the maximum of $n=1000$ products of Gamma$(2,1)$ and $\text{U}(0,1)$ random variables. For comparison, the density of Gamma$(2,1)$, $f_{\Gamma}(x)=x\exp{(-x)}$, is also plotted.}}
\label{gamma1000fig}
\end{figure}

\subsection{Cumulative Size of the Most Important Files}
Here we use the same assumptions as in the previous section for the file size and popularity distributions, but instead of only finding the size of the most important file, we are now interested in finding the cumulative size of $q$ most important files. For example, out of a catalogue of size $n=100$ files, we can then find the expected size of $q=50$ most important files.

We use the expected value of the asymptotic factor distribution \eqref{afd} to find an approximation of the expected size of the $k^{\text{th}}$ least important file with $k\in[n-q+1,n-1]$ and $q>1$. Recall that in the parlance of order statistics, the $k^{\text{th}}$ order statistic corresponds to the $k^{\text{th}}$ smallest variate. Thus, to find the expected sizes of $q$ of the most important files in an $n$-file catalogue, we must study the cases where $k\in[n-q+1,n]$. Also note that as it is clear that the asymptotic result cannot be used for $k=n$, here we simply study the cases $k\in[n-q+1,n-1]$ and use \eqref{maxtraffic} to find the exact expected size of the most important file.

As shown in the previous example, the cumulative distribution of the product of Gamma$(2,1)$ and U$(0,1)$ random variables is $G_2(x) = F_U(x) = 1-\exp{(-x)}$. The inverse of this is $G_2^{-1}(x) = -\log(1-x)$ which implies $\beta(\eta) = \beta(k/n) = G_2^{-1}(\eta) = -\log(1-\eta) = -\log(1-k/n)$. Therefore, we get the limiting distribution
\begin{align}\nonumber
\lim\limits_{k/n=\eta, n\to\infty} f_{\tilde{X}}(x) = \frac{1}{\exp{(-\beta)}x}x\exp{\left(-x\right)} = \exp{(\beta-x)},
\end{align}
where $x\in[\beta,\infty)$. The expected value related to this distribution yields
\begin{align}\label{Ekn}
\mathcal{E}(k,n) = \int_{\beta(k,n)}^\infty x\exp{(\beta(k,n)-x)} {\rm d}x = 1-\log(1-k/n).
\end{align}
%

Finally, we can now find an approximation of the sum of the expected file sizes when $k\in [n-q+1,n-1]$. This sum becomes
\begin{align}
\mathcal{S}'(q,n)=\sum_{i=n-q+1}^{n-1}\mathcal{E}(i,n) = -\log \left(-\left(-\frac{1}{n}\right)^q n (1-q)_{q-1}\right)+q-1,
\end{align}
where $(a)_n$ is the Pochhammer symbol.

Now let $\mathcal{S}(q,n)$ denote the exact value of the sum of the expected sizes of $q$ most important files of an $n$-file catalogue. With \eqref{maxtraffic}, we get
\begin{align}\label{Sqneqn}
      \mathcal{S}(q,n) &\approx \mathcal{E}_\text{max}(n) + \mathcal{S}'(q,n)\nonumber\\ &= \frac{1}{n}+H_ {n-1}+q-\log \left(-\left(-\frac{1}{n}\right)^q n (1-q)_{q-1}\right).
\end{align}
Since a special case of the Pochhammer symbol yields $(0)_0=1$, \eqref{Sqneqn} is valid for all $q\in [1,n]$, and especially $\mathcal{E}_\text{max}(n) = \mathcal{S}(1,n)$.

As an example of cumulative cache size requirements, Fig. \ref{percentageqfig} plots
\begin{align}\label{Reqn}
R(q) = \frac{\mathcal{S}(q,100)}{\mathcal{S}(100,100)},
\end{align}
\emph{i.e.}, the ratio of the approximated cache space consumption of $q$ of the most important files and the approximated whole expected size of a catalogue of $n=100$ files. We see that, for example, caching $32$ most important files already requires a cache that is more than half of the size of the whole catalogue. We also see that the approximation \eqref{Reqn} is well in line with Monte Carlo simulations.
\begin{figure}[!h]
\centering
\includegraphics[scale=.4]{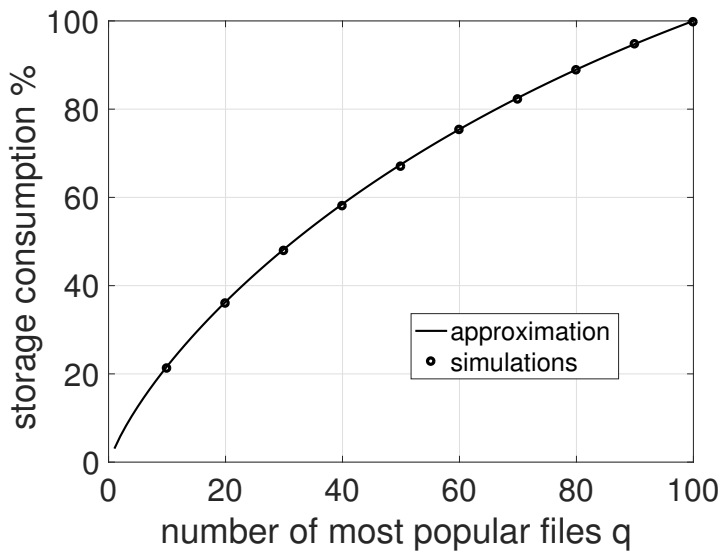}
\caption{Cache space consumption of the most important files.}
\label{percentageqfig}
\end{figure}

\section{Conclusions and Future Work}\label{conclusionssec}
We have derived general closed-form expressions for the densities of the addends and factors of ordered sums and products. Furthermore, we have presented the corresponding asymptotic distributions. Through concrete examples, we have shown some interesting properties of such distributions. As an application to caching, we have derived the size distribution of the file that generates most data traffic, along with tools for approximating cumulative file sizes of a certain number of the most important files by finding the expected values of the asymptotic distributions.

Natural extensions of this work include finding the addend and factor distributions of dependent random variables, as well as expressions for specific distributions with an arbitrary number of addends or factors. Further, from the aspect of wireless communication, especially factor distributions of various channel fading models are of special interest. Also, the harmonic mean, as used in end-to-end SNR calculations, would provide an attractive case for future studies.

\section*{Acknowledgements}
The authors would like to thank Ejder Ba\c{s}tu\u{g} for fruitful discussions.

\begin{appendices}\label{appendixsec}

\section{Proof of factor density of product of normal distributions with $m=n=k=2$}\label{norfac}

Let us focus on the following double integral
\begin{align*}
f_{\tilde{X}}(x) = 2\phi(x)\int_{-\infty}^\infty \frac{1}{\pi} \int_{-\infty}^{zx} K_0(\abs t){\rm d}t\phi(z){\rm d}z.
\end{align*}
Since $x$ can take both positive and negative values, let us first consider the case $x\geq 0$. Noting the fact that $z$ can either be positive or negative, we decompose the above double integral, for $x\geq 0$, as

\begin{align*}
f_{\tilde{X}}(x) = 2\phi(x)\int_{-\infty}^0 \frac{1}{\pi} \int_{-\infty}^{-\abs z x} K_0(\abs t){\rm d}t\phi(z){\rm d}z + 2\phi(x)\int_0^{\infty} \frac{1}{\pi} \int_{-\infty}^{zx} K_0(|t|){\rm d} t\phi(z){\rm d}z.
\end{align*}
Since $K_0(|t|)$ is an even function, we can further simplify the inner integral to yield
\begin{align*}
 f_{\tilde{X}}(x) &=2\phi(x) \left\{\int_{-\infty}^0 \left(\frac{1}{2} - \int_{-|z|x}^0 K_0(|t|){\rm d}t\right)\phi(z){\rm d}z + \int_0^\infty\left(\frac{1}{2}+\frac{1}{\pi}\int_{0}^{zx}K_0(|t|){\rm d}t\right)\phi(z){\rm d}z\right\}\\
 &= \phi(x)-\frac{2\phi(x)}{\pi}\int_{-\infty}^0 \int_{0}^{|z|x}K_0(|t|)\phi(z)\mathop{dz}+\frac{2\phi(x)}{\pi}\int_0^\infty \int_0^{zx}K_0(|t|){\rm d}t\phi(z){\rm d}z.
\end{align*}
The next key step is to observe that $\int_{0}^{|z|x}K_0(|t|)$ and $\phi(z)$ are even functions of $z$. Therefore, we have, for $x\geq 0$,
\begin{align*}
f_{\tilde{X}}(x) & = \phi(x)-\frac{2\phi(x)}{\pi}\int_{0}^\infty \int_{0}^{|z|x}K_0(|t|){\rm d}t\phi(z){\rm d}z+\frac{2\phi(x)}{\pi}\int_0^\infty \int_0^{zx}K_0(|t|){\rm d}t\phi(z){\rm d}z.
\end{align*}
Noting the fact that $\int_0^\infty\int_0^{zx}K_0(|t|){\rm d}t\phi(z){\rm d}z<\infty$, we obtain
\begin{align*}
f_{\tilde{X}}(x) & =\phi(x), \;\;\;x\geq 0.
\end{align*}
The case corresponding to $x<0$ can easily be proved in a similar manner and is thus omitted.

\section{Proof of double scaling limit of addend density}\label{dslproof}
 Let us rewrite (\ref{eq_prop1}) with a slightly modified notation as
\begin{align*}
f_{\tilde{X}}(x) = \frac{n!}{(k-1)! (n-k)!}f_X(x) \mathbb{E}_T\left\{p^{k-1}(1-p)^{n-k}\right\}
\end{align*}
where $p=F_m(x+t)$ with $t=\sum_{j=2}^mx_j$ and $g_{m-1}(t)$ denotes the probability density of $T$. Therefore, we obtain, after some algebraic manipulations
\begin{align}
f_{\tilde{X}}(x) &= 
\frac{\Gamma(n+1)}{\Gamma(k)\Gamma(n-k+1)}f_X(x) \int_{-\infty}^{\infty}F_m^{k-1}(y)[1-F_m(y)]^{n-k}g_{m-1}(y-x)\mathop{dy}\nonumber\\
&= \frac{\Gamma(n+1)}{\Gamma(k)\Gamma(n-k+1)}f_X(x) \int_{-\infty}^{\infty}F_m^{k-1}(y)[1-F_m(y)]^{n-k}g_m(y)\frac{g_{m-1}(y-x)}{g_m(y)}\mathop{dy}.\nonumber
\end{align}
To facilitate further analysis, let us use the variable transformation $z=F_m(y)$ with $dy=g_m(y) dy$ to arrive at
\begin{align*}
f_{\tilde{X}}(x)
&=f_X(x)\int_0^1 \frac{g_{m-1}(F_m^{-1}(z)-x)}{g_{m}(F_m^{-1}(z))}\mathop{dp_{n,k}(z)}
\end{align*}
where 
\begin{align*}
\mathop{dp_{n,k}(z)} =  \frac{\Gamma(n+1)}{\Gamma(k)\Gamma(n-k+1)} z^{k-1}(1-z)^{n-k}\mathop{dz}
\end{align*}
denotes a beta distribution with parameters $k$ and $n-k+1$.
Now the desired double scaling limit can be expressed as
\begin{align*}
\lim_{n,k\to \infty, \frac{k}{n}=\eta} f_{\tilde{X}}(x)
=f_X(x)\lim_{n\to\infty} \int_0^1 \frac{g_{m-1}(F_m^{-1}(z)-x)}{g_{m}(F_m^{-1}(z))}\mathop{dp_{n,\eta n}(z)}.
\end{align*}
The random variable $Z_n$ associated with the beta density function $p_{n,\eta n}(z)$ has the characteristic function 
\begin{align*}
\phi_n{(\omega)}&=\int_0^1 \frac{\Gamma(n+1)}{\Gamma(k)\Gamma(n-k+1)} z^{k-1}(1-z)^{n-k}\exp(z\omega {\rm j})\alert{\mathop{dz}}\\
&={}_1F_1(k;n+1;-j\omega),
\end{align*}
where ${}_1F_1$ is the confluent hypergeometric function of the first kind and ${\rm j}=\sqrt{-1}$. Moreover, clearly, $\phi_n{(\omega)}$ converges $\forall \omega \in (-\infty,\infty)$ and 
\begin{align*}
\phi_{\infty}(\omega)= \exp({\rm j}\omega \eta)
\end{align*}
with $\phi_{\infty}(\omega)$ being continuous at $\omega=0$. Note that the limiting characteristic function corresponds to a point mass at $z=\eta$. This along with the fact that $\frac{g_{m-1}[F_m^{-1}(z)-x]}{g_m[F_m^{-1}(z)]}$ is bounded and continuous for $z\in(0,1)$ enables us to invoke \cite[Theorems 6.3.2 and 4.4.2]{Chung} to yield
\begin{align*}
\lim_{n,k\to \infty, \frac{k}{n}=\eta} f_{\tilde{X}}(x)&=f_X(x)\lim_{n\to\infty} \int_0^1 \frac{g_{m-1}(F_m^{-1}(z)-x)}{g_{m}(F_m^{-1}(z))}\mathop{dp_{n,\eta n}(z)}\nonumber\\
&=f_X(x)\int_0^1\frac{g_{m-1}[F_m^{-1}(z)-x]}{g_m[F_m^{-1}(z)]}\delta(z-\eta)\mathop{dz}
\end{align*}
where $\delta(z)$ is the Dirac delta function. Finally, using the sifting property of the delta function yields the desired result.

\section{Proof of double scaling limit of factor density}\label{dslfacnor}
 Let us rewrite (\ref{eq_prop2}) with a slightly modified notation as
\begin{align*}
f_{\tilde{X}}(x)=\frac{n!}{(n-1)!(n-k)!}f_X(x)\mathbb{E}_V\left\{q^{k-1}(1-q)^{n-k}\right\}
\end{align*}
where $q=F_m(xv)$ with $v=\Pi_{j=2}^nx_j$ and $g_{m-1}(v)$ denotes the pdf of $V$. Therefore, we have 
\begin{align*}
f_{\tilde{X}}(x)
&=\frac{n!}{(n-1)!(n-k)!}f_X(x)\int_{-\infty}^{\infty}\{F_m^{k-1}(xv)\left(1-F_m(xv)\right)^{n-k}g_{m-1}(v)\mathop{dv}
\end{align*}
from which we obtain after introducing the transformation $xv=y$
\begin{align*}
f_{\tilde{X}}(x)&=\frac{n!}{(n-1)!(n-k)!}f_X(x)\int_{-\infty}^{\infty}\{F_m^{k-1}(y)\left(1-F_m(y)\right)^{n-k}\frac{g_{m-1}(y/x)}{|x|}\mathop{dy}\\
&=\frac{n!}{(n-1)!(n-k)!}f_X(x)\int_{-\infty}^{\infty}\{F_m^{k-1}(y)\left(1-F_m(y)\right)^{n-k}g_m(y)\frac{g_{m-1}(y/x)}{xg_m(y)}\mathop{dy}.
\end{align*}
Following similar arguments as before, we can rewrite the above as
\begin{align*}
\lim_{n,k\to \infty, \frac{k}{n}=\eta} f_{\tilde{X}}(x)=f_X(x) \int_0^1\frac{g_{m-1}(y/x)}{|x|g_m(y)} dp_{n,k}(z).
\end{align*}
Now one can closely follow the steps shown in the previous proof to arrive at the desired answer.

\end{appendices}
\newpage
\bibliographystyle{ieeetr}

\begin{thebibliography}{10}

\bibitem{cisco}
Cisco, ``Cisco Visual Networking Index: Forecast and Methodology, 2015--2020," White Paper, http://www.cisco.com/c/en/us/solutions/collateral/service-provider/visual-networking-index-vni/complete-white-paper-c11-481360.pdf, 2016.

\bibitem{belady1966}
L. A. Belady, ``A Study of Replacement Algorithms for a Virtual-Storage Computer," \emph{IBM Syst. J.}, vol. 5, no. 2, pp.~78--101, 1966.

\bibitem{wangoct1999}
J. Wang, ``A Survey of Web Caching Schemes for the Internet," \emph{SIGCOMM Compt. Commun. Review,} pp.~36--46, vol. 29, no. 5, Oct. 1999.

\bibitem{podlipnig2003}
S. Podlipnig and L. B\"osz\"ormenyi, ``A Survey of Web Cache Replacement Strategies," \emph{ACM Comput. Surveys (CSUR),} vol. 35, no. 4, pp.~374--398, 2003.

\bibitem{borst2010}
S. Borst, V. Gupta, and A. Walid, ``Distributed Caching Algorithms for Content Distribution Networks," in \emph{Proc. IEEE INFOCOM}, pp.~1--9, Mar. 2010.


\bibitem{bastug}%
E. Ba\c{s}tu\u{g}, M. Bennis, and M. Debbah, ``Living on the Edge: The Role of Proactive Caching in 5G Wireless Networks," \emph{IEEE Commun. Magazine}, vol.~52, no.~8, pp.~82--89, Aug. 2014.

\bibitem{gregori}%
M. Gregori, J. Gomez-Vilardebo, J. Matamoros, and D. G\"und\"uz, ``Wireless Content Caching for Small Cell and D2D Networks," \emph{IEEE J. Sel. Areas Commun.}, vol.~34, no.~5, pp.~1222--1234, May 2016.

\bibitem{golrezaeibase}%
N. Golrezaei, P. Mansourifard, A. F. Molisch, and A. G. Dimakis, ``Base Station Assisted Device-to-Device Communications for High-Throughput Wireless Video Networks," \emph{IEEE Trans. Wireless Commun.}, vol.~13, no.~7, pp.~3665--3676, Jul. 2014.

\bibitem{maddah}%
M. A. Maddah-Ali and U. Niesen, ``Fundamental Limits of Caching," \emph{IEEE Trans. Inf. Theory}, vol.~60, no.~5, pp.~2856--2867, May 2014.

\bibitem{song}%
J. Song, H. Song, and W. Choi, ``Optimal Caching Placement of Caching System with Helpers," in \emph{Proc. IEEE Int. Conf. Commun. (ICC)}, pp.~1825--1830, Jun. 2015.

\bibitem{ji2016}%
M. Ji, G. Caire, and A. Molisch. ``Fundamental Limits of Distributed Caching in D2D Wireless Networks," in \emph{Proc. IEEE Inf. Theory Wksp. (ITW)}, pp.~1--5, Sept. 2013.

\bibitem{tamoorjune2016}
S. Tamoor-ul-Hassan, M. Bennis, P. H. J. Nardelli, and M. Latva-aho, ``Caching in Wireless Small Cell Networks: A Storage-Bandwidth Tradeoff," \emph{IEEE Commun. Lett.}, vol. 20, no. 6, pp.~1175--1178.

\bibitem{osmanmay2011}%
N. I. Osman, T. El-Gorashi, and J. M. H. Elmirghani, ``Reduction of Energy Consumption of Video-on-Demand Services Using Cache Size Optimization," in \emph{Proc. Wireless and Opt. Commun. Netw. (WOCN)}, pp.~1--5, May 2011.

\bibitem{zhaioptimaldec2011}%
H. Zhai, A. K. Wong, H. Jiang, Y. Sun, and J. Li, ``Optimal P2P Cache Sizing: A Monetary Cost Perspective on Capacity Design of Caches to Reduce P2P Traffic," in \emph{Proc. Parallel and Distrib. Syst. (ICPADS)}, pp.~565--572, Dec. 2011.

\bibitem{liumay2016}
A. Liu and V. K. N. Lau, ``How Much Cache is Needed to Achieve Linear Capacity Scaling in Backhaul-Limited Dense Wireless Networks?," \emph{IEEE/ACM Trans. on Netw.}, no. 99, May 2016.

\bibitem{zhang2015}
J. Zhang, X. Lin, C. C. Wang, and X. Wang, ``Coded Caching for Files with Distinct File Sizes", in \emph{Proc. IEEE Int'l. Symp. on Inf. Theory (ISIT)}, pp.~1686--1690, Jun. 2015.

\bibitem{maggi}
L. Maggi, L. Gkatzikis, G. Paschos, and J. Leguay, ``Adapting Caching to Audience Retention Rate: Which Video Chunk to Store?," arXiv:1512.03274, 2015.

\bibitem{cherkasova1998}
L. Cherkasova, ``Improving WWW Proxies Performance with Greedy-Dual-Size-Frequency Caching Policy," in \emph{HP Technical Report}, Palo Alto, 1998.

\bibitem{ali2011}
W. Ali, S. M. Shamsuddin, and A. S. Ismail, ``A Survey of Web Caching and Prefetching," \emph{Int'l. J. Advances Soft Comput. Appl.}, vol. 3, no. 1, pp.~2074--8523, Mar. 2011.

\bibitem{zhang2012}
J. Zhang, ``A Literature Survey of Cooperative Caching in Content Distribution Networks," arXiv:1210.0071, 2012.

\bibitem{niesen2014}
U. Niesen and M. A. Maddah-Ali, ``Coded Caching with Nonuniform Demands," in \emph{Proc. IEEE INFOCOM}, pp.~221--226, Apr.--May 2014.

\bibitem{romano2011}
S. Romano and H. El Aarag, ``A Neural Network Proxy Cache Replacement Strategy and its Implementation in the Squid Proxy Server," \emph{Neural Comput. and Appl.}, vol. 20, no. 1, pp.~59--78, Feb. 2011.

\bibitem{ali2009}
W. Ali and S. M. Shamsuddin, ``Intelligent Client-Side Web Caching Scheme Based on Least Recently Used Algorithm and Neuro-Fuzzy System," in \emph{Proc. Int'l. Symp. Neural Netw.}, pp.~70--79, May 2009.

\bibitem{abdalla2015}
A. Abdalla, S. Sulaiman, and W. Ali, ``Intelligent Web Objects Prediction Approach in Web Proxy Cache Using Supervised Machine Learning and Feature Selection," in \emph{Int'l. J. Advances Soft Compu. Appl.,} vol. 7, no. 3, pp.~2074--8523, Nov. 2015.

\bibitem{olmos}
F. Olmos, B. Kauffmann, A. Simonian, and Y. Carlinet, ``Catalog Dynamics: Impact of Content Publishing and Perishing on the Performance of a LRU Cache," in \emph{Proc. Int'l. Teletraffic Congress}, Sept. 2014.

\bibitem{traverso}
S. Traverso, M. Ahmed, M. Garetto, P. Giaccone, E. Leonardi, and S. Niccolini, ``Unravelling the Impact of Temporal and Geographic Locality in Content Caching Systems," \emph{IEEE Trans. Multimedia}, vol.~7, no.~3, pp.~1839--1854, Sept. 2015.

\bibitem{fagin}
R. Fagin and T. G. Price. ``Efficient Calculation of Expected Miss Ratios in the Independent Reference Model," \emph{SIAM J. Comput.}, vol.~7, no.~3, pp.~288--297, 1978.

\bibitem{bharath}
B. N. Bharath, K. G. Nagananda, and H. V. Poor,  ``A Learning-Based Approach to Caching in Heterogenous Small Cell Networks," \emph{IEEE Trans. Commun.}, vol.~64, no.~4, pp.~1674--1686, Apr. 2016.

\bibitem{clark}
C.~E. Clark, ``The Greatest of a Finite Set of Random Variables,'' \emph{Operations Research}, vol.~9, pp.~145--162, Mar. 1961.

\bibitem{david}
H.~A. David and H.~N. Nagaraja, {\em Order Statistics}.
\newblock Wiley, 2005.

\bibitem{gungor}
M.~G{\"u}ng{\"o}r, Y.~Bulut, and S.~{\c C}alik, ``Distributions of Order Statistics,'' {\em Appl. Math. Sciences}, vol.~3, pp.~795--802, Apr. 2009.

\bibitem{wilks}
S.~S. Wilks, ``Order statistics,'' {\em Bulletin of the American Mathematical Society}, vol.~54, no.~1, pp.~6--50, 1948.

\bibitem{donahue1964}
J.~D. Donahue, {\em Products and Quotients of Random Variables and Their Applications}.
\newblock Government Publication, 1964.

\bibitem{rohat1976}
V.~K. Rohatgi, {\em An Introduction to Probability Theory Mathematical Statistics}.
\newblock John Wiley and Sons, New York, 1976.

\bibitem{springer1979}
M.~D. Springer, {\em The Algebra of Random Variables}.
\newblock Wiley, 1979.

\bibitem{cook1981}
J. I. D. Cook, ``The H-function and Probability Density Functions of Certain Algebraic Combinations of Independent Random Variables with H-function Probability Distribution," Ph.D. dissertation, The University of Texas at Austin, Austin, TX, May 1981. 

\bibitem{mixed}
J.~M. Wozencraft and I.~M. Jacobs, {\em Principles of Communication Engineering}.
\newblock Wiley, 1965.

\bibitem{middleton}
D. Middleton, \emph{An Introduction to Statistical Communication Theory.} New York: McGraw-Hill, 1960.

\bibitem{PD19}
P. Dharmawansa, N. Rajatheva, and K. Ahmed, ``On the Distribution of the Sum of Nakagami-m Random Variables," \emph{IEEE Trans. Commun.}, vol.~55, no.~7, pp.~1407--1416, Jul. 2007.

\bibitem{weissuni}
E. W. Weisstein, ``Uniform Product Distribution." From MathWorld--A Wolfram Web Resource. http://mathworld.wolfram.com/UniformProductDistribution.html

\bibitem{Chung}
K. L. Chung, \emph{A Course in Probability Theory}, 3rd ed., New York: Academic Press, 2001.

\bibitem{weiss}
E. W. Weisstein, ``Normal Product Distribution." From MathWorld--A Wolfram Web Resource. http://mathworld.wolfram.com/NormalProductDistribution.html

\bibitem{karakacha}
K. K. Karakacha, ``Exponential Distribution: its Constructions, Characterizations and Related Distributions ," Master's thesis, School of Mathematics, University of Nairobi, Nairobi, Kenya, May 2009. 

\bibitem{gradshteyn}
I. S. Gradshteyn and I. M. Ryzhik, \emph{Table of Integrals, Series, and Products.} New York: Academic, 1980.

\bibitem{johnson}
N. L. Johnson, S. Kotz, and N. Balakrishnan. {\em Continuous Univariate Distributions}, vol. 2, 2nd edition. Wiley. p. 306, 1995. ISBN 0-471-58494-0.

\end{thebibliography}

\end{document}